\documentclass{amsart}

\usepackage[pagewise]{lineno}
\usepackage{amsthm,amssymb,verbatim,amsmath, amsfonts, amscd}
\usepackage{xcolor}
\usepackage{amssymb}
\usepackage{amsthm}

\usepackage{multirow}
\usepackage{array}
\usepackage{amsthm,amssymb,verbatim,amsmath, amsfonts, amscd
,a4wide
,color}
\usepackage{ascmac}

\usepackage{verbatim}

\theoremstyle{plain}\newtheorem{Theorem}{Theorem}[section]
\theoremstyle{plain}\newtheorem{Conjecture}[Theorem]{Conjecture}
\theoremstyle{plain}\newtheorem{Corollary}[Theorem]{Corollary}
\theoremstyle{plain}\newtheorem{Lemma}[Theorem]{Lemma}
\theoremstyle{plain}\newtheorem{Proposition}[Theorem]{Proposition}
\theoremstyle{definition}\newtheorem{Definition}[Theorem]{Definition}
\theoremstyle{definition}\newtheorem{Example}[Theorem]{Example}
\theoremstyle{definition}
\theoremstyle{definition}
\theoremstyle{definition}
\theoremstyle{definition}\newtheorem{Remark}[Theorem]{Remark}
\theoremstyle{definition}
\theoremstyle{definition}\newtheorem{Question}[Theorem]{Question}

\def\PAut{\mathrm{PAut}}

\def\dim{\mathrm{dim}}

\newcommand{\FF}{{\mathbb{F}}}
\newcommand{\bigzero}{\mbox{\normalfont\Large\bfseries 0}}

\begin{document}

\title{On the permutation automorphisms of binary cubic codes}

\author{{Murat Altunbulak, Fatma Altunbulak Aksu, Roghayeh Hafezieh and {\.I}pek Tuvay}}
\address{DOKUZ EYL\"{U}L UNIVERSITY, DEPARTMENT OF MATHEMATICS, 35390, BUCA, \.{I}ZM\.{I}R, T\"{U}RK\.{I}YE}
\email{murat.altunbulak@deu.edu.tr}
\address{M\.{I}MAR S\.{I}NAN FINE ARTS UNIVERSITY, DEPARTMENT OF MATHEMATICS, 34380, BOMONT\.{I}, \c{S}\.{I}\c{S}L\.{I}, \.{I}STANBUL, T\"{U}RK\.{I}YE}
\email{fatma.altunbulak@msgsu.edu.tr}
\address{KOCAEL\.{I}, T\"{U}RK\.{I}YE}
\email{r.hafezieh@yahoo.com}
\address{M\.{I}MAR S\.{I}NAN FINE ARTS UNIVERSITY, DEPARTMENT OF MATHEMATICS, 34380, BOMONT\.{I}, \c{S}\.{I}\c{S}L\.{I}, \.{I}STANBUL, T\"{U}RK\.{I}YE}
\email{ipek.tuvay@msgsu.edu.tr}

\thanks{}

\keywords{binary linear code, fixed point free permutation automorphism, fixed subcode associated to a permutation automorphism}
\subjclass[2020]{94B05, 11T71, 20B05}

\begin{abstract}

 A binary linear code whose permutation automorphism group has a fixed point free permutation of order $3$ is called a binary cubic code. The scope of
 this paper is to investigate the structural properties of binary cubic codes. Let $C$ be a binary cubic $[n,k]$ code. In this paper, we prove that if $n\geq 30$ and $C$ has permutation automorphism group of order three, then $k\geq 6$. Additionally, we show that if  $n < 30$ and $k\leq 4$, then the permutation automorphism group of $C$ has order greater  than three.   Moreover, along the way, we provide some results on the structure of the  higher dimensional cubic codes. In particular, we present some
 results concerning the structure of the putative extremal self-dual $[72,36,16]$ code under the assumption that it is cubic.

\end{abstract}

\maketitle

\section{Introduction}

Given a binary linear code $C$, let $\PAut(C)$ denote its permutation automorphism group. Our motivation in this paper comes from the following question
mentioned in ~\cite{Borello2015}:

\begin{Question}
Let $C$ be a (self-dual) binary linear code and suppose that $\PAut(C)$ contains a non-trivial subgroup $G$. What can we say about $C$ knowing $G$?
\end{Question}

 A well-known result states that whenever $G$ is a non-trivial subgroup of ${\rm{PAut}}(C)$, then $C$ becomes a module over the group algebra
 ${\mathbb{F}_2}G$. This leads to study linear codes by means of different algebraic tools which is an active line of research, see for instance
 \cite{BW2022}. Indeed,  the structure or the existence of $C$ can be determined by focusing on some specific automorphisms in ${\rm{PAut}}(C)$, (see
 \cite{AksuHafeziehTuvay2023,BW2013,Yorgov1983}). In \cite{AksuHafeziehTuvay2023}, the authors considered the properties of binary linear codes of even
 length whose permutation automorphism group is a cyclic group generated by an involution. They prove that, up to dimension or co-dimension $4$, there
 is no quasi-group code whose permutation automorphism group is isomorphic to $C_2$.

  A binary linear code whose permutation automorphism group has a fixed point free permutation of order $3$ is called a binary cubic code in the
  literature (see \cite{BDNS2003}). It should be mentioned that, a famous problem in coding theory discusses the existence of a putative self-dual
  $[72,36,16]$ code. This problem comes down to the study of its permutation automorphism group for which the possibilities are $1, C_2, C_3, C_2 \times
  C_2, C_5$ (see \cite{Borello2015} for a summary of this problem).  Besides, in \cite{SB2004}, it is proved that a putative self-dual $[72,36,16]$ code
  has no automorphism of order $3$ with fixed points. In consequence, if a putative self-dual $[72,36,16]$ code, whose permutation automorphism group is
  $C_3$, exists, then it is a cubic code.

  Inspired by these results, in this paper, we consider the permutation automorphism groups of cubic codes. We first consider low dimensional cubic
  codes and prove the following theorems.

\begin{Theorem}\label{main1}  Let $C$ be a binary cubic $[n,k]$ code, where $n\geq 30$. If the order of the permutation automorphism group of $C$ is
$3$, then $k\geq 6$.

\end{Theorem}

\begin{Theorem}\label{main2} Let $C$ be a binary cubic $[n,k]$ code, where $n<30$. If $k\leq 4$, then the order of the
permutation automorphism group of $C$ is greater than $3$.

\end{Theorem}
\begin{Remark}

In \cite{LingSole}, it is proved that all cubic codes can be obtained by a generalized cubic construction from a binary code and a quaternary code, both
of the same length. Using the construction methods mentioned in \cite{LingSole} (see Theorem 5.1 and Section VI, case $m=3$ and Turyn's Construction),
we obtain some computational results for $5$-dimensional binary cubic codes of length less than $30$. Indeed, we consider cubic codes of dimension $5$
and among them, we determine the smallest order for the permutation automorphism groups of cubic codes of length $n$, where $n\in\{6,9,12,15,18,21\}$, in
SAGEMATH. For $n\in \{24, 27\}$, the calculations in SAGEMATH takes too much time, but for these lengths we have partial results which support the claim
that a $5$-dimensional cubic code of length $n\in \{24, 27\}$ has a permutation automorphism group whose order is greater than $3$. (for detailed
explanations see Remark \ref{Calc}).
\end{Remark}
Moreover, we construct some examples of binary cubic codes of dimension greater than or equal to $ 6$ having permutation automorphism groups of order $3$. This verifies that
five is a possibly strict lower bound for the dimension of the cubic code $C$ to have a permutation automorphism group of order $3$. Generalizing all
the methods to prove the results for cubic linear codes of dimension less than or equal to $5$, we obtain some results on the  structure of cubic codes
for higher dimensions. In particular, we obtain some results on the structure of a putative self-dual $[72,36,16]$ code having permutation automorphism
group of order $3$.

The structure of the paper is as follows. In Section 2, we introduce some notations and discuss lower dimensional cases. In Section 3, we present the
proofs of the main results. In Section 4, we give some general methods for higher dimensional cubic codes. We construct examples of higher dimensional
cubic codes having permutation automorphism group of order $3$. In Section 5, we presents the applications on putative extremal self-dual $[72, 36, 16]$
code. We end the paper with a conjecture which states that the automorphism groups of  a putative extremal self-dual $[72, 36, 16]$ code is not
isomorphic to $C_3$.

\section{Notation and results for lower dimensions}

A binary linear code of length $n$ and dimension $k$ is a $k$-dimensional subspace of the vector space $\FF_2^n$. The symmetric group $S_n$ acts on
$\FF_2^n$ as follows. For $\sigma\in S_n$ and for $v=v_1v_2\dots v_n$, we define $v^{\sigma}=v_{\sigma^{-1}(1)}v_{\sigma^{-1}(2)}\dots
v_{\sigma^{-1}(n)}$.
A permutation automorphism of $C$ is a permutation $\sigma\in S_n$ such that $C^\sigma=C$. The set of all automorphisms of $C$ forms a subgroup of
$S_n$, which is called the permutation automorphism group of $C$ and it is denoted by ${\rm PAut}(C)$.

Let $C$ be a binary linear code. Assume that  $\sigma$ is a fixed point free permutation of order $3$ in ${\rm  PAut}(C)$. In this case, it is clear
that $C$ has length a multiple of $3$. So for the rest of the paper, we assume that $C$ is a binary linear code of length $n=3m$ where $m$ is  a
positive integer.
Since any conjugate of a permutation automorphism of $C$ is a permutation automorphism of a code which is permutation equivalent to $C$, without loss of
generality, we can assume that $$\sigma=(1,2,3)(4,5,6)...(n-2,n-1,n).$$
Let $\sigma_i$ be the $i$th cycle of $\sigma$ so that $\sigma=\prod_{i=1}^m \sigma_i$ and let
$\Omega_i$ be the $i$th cycle set of $\sigma$ for $i=1, \ldots, m$.
Given $v$ in $C$, let $v|_{\Omega_{i}}$ denote the $i$th $3$-block of $v$ associated to the $i$th cycle set $\Omega_{i}$.
The subcode $F_{\sigma}(C)=\{v \in C \ | \ v^{\sigma}=v\}$ is called the fixed subcode of $C$ associated to $\sigma$ and there is
another subcode which is defined as follows
$$E_{\sigma}(C)=\{v \in C \ | \ \sum_{i\in \Omega_j} v_i=0 \text{ for any } j=1, \ldots m \}.$$ It is clear that both $F_{\sigma}(C)$ and
$E_{\sigma}(C)$ are $ \sigma$-invariant.

The celebrated work of Cary Huffman gives a decomposition in terms of these subcodes associated to an automorphism.

\begin{Theorem} \cite[Theorem 1]{ Huffman}  With the notations above, we have that $C=F_{\sigma}(C)\oplus E_{\sigma}(C)$.

\end{Theorem}

Moreover, Huffman proves the following result which plays an important role in our discussion.

\begin{Lemma}\cite[Corollary 1]{Huffman}~\label{cor:1}
Let $C$ be a linear code whose base field is $\mathbb{F}_{q}$ and let $\sigma$ be a permutation
in ${\rm PAut}(C)$. If the order of $\sigma$ is a power of a prime number $s$, then the multiplicative
order of $q$ modulo $s$, divides $\dim (E_{\sigma}(C))$.
\end{Lemma}

Let $C$ be a linear code whose permutation automorphism group contains a fixed point free permutation $\sigma$ of order $3$. Such a code is called cubic
(see \cite{BDNS2003,LingSole}). Moreover, Lemma ~\ref{cor:1} implies that $2$ divides ${\rm dim}(E_{\sigma}(C))$. We start our discussion introducing
some notations which will be used in our proofs.

Given a non-zero vector $v$ in $E_{\sigma}(C)$, assume that $v|_{\Omega_{i}}$ is a non-zero $3$-block of $v$, then it follows from the definition of
$E_\sigma(C)$ that there are exactly two $1$'s in this block. Let $a_{i,v,1}$, $a_{i,v,2}$, and $a_{i,v,3}$ denote the positions of the first $1$, the
second $1$, and zero in the $3$-block $v|_{\Omega_{i}}$, respectively. Moreover, let ${\rm J}_{v}$ be the set of indices $i$ for which $v|_{\Omega_{i}}$
is a non-zero $3$-block.
Assume $0\neq w\in E_{\sigma}(C)$ is a vector different from $v$ whose weight is equal to that of $v$. If $wt(v)=wt(w)=2m$, then this implies that
neither $v$ nor $w$ has a zero $3$-block. If $wt(v)=wt(w)=2t < 2m$, observe that $|{\rm J}_{v}\setminus {\rm J}_{w}|=|{\rm J}_{w}\setminus {\rm
J}_{v}|$. Let $d:=|{\rm J}_{v}\setminus {\rm J}_{w}|=|{\rm J}_{w}\setminus {\rm J}_{v}|$. If $ {\rm J}_{v}\setminus {\rm J}_{w}\neq \emptyset$, then
$d\geq 1$. Moreover  if $i\in{\rm J}_{v}\setminus {\rm J}_{w}$, then  $v|_{\Omega_{i}}$ is a non-zero $3$-block while $w|_{\Omega_{i}}$ is a zero
$3$-block. Let $${\rm J}_{v}\setminus {\rm J}_{w}:=\{r_{w,1},...,r_{w,d}\}$$ and $${\rm J}_{w}\setminus {\rm J}_{v}:=\{r_{v,1},...,r_{v,d}\}$$
    where for each $j\in\{1,...,d\}$, the symbols $r_{v,j}$ denote the $j$th zero $3$-block in $v$ whose associated $3$-block in $w$ is non-zero.

\begin{Example}~\label{example:1} Let $\sigma=(1,2,3)(4,5,6)(7,8,9)(10,11,12)(13,14,15)(16,17,18)$ and
 $C$  a binary linear code spanned by the following vectors:
$$v=110 000 101 011 000 000,$$ $$w=000 101 110 000 000 101,$$ $$v^{\sigma}=011 000 110 101 000 000,$$ $$w^{\sigma}=000 110 011 000 000 110.$$
 It is clear that $wt(v)=wt(w)=6< 2m=12$. Moreover, we have that  $\sigma \in {\rm PAut}(C)$ and $C=E_{\sigma}(C)$.
 On the other hand $J_{v}\backslash J_{w}=\{1, 4\}$, $J_{w}\backslash J_{v}=\{2,6\}$, $r_{w,1}=1$, $r_{w,2}=4$, $a_{r_{w,1},v,1}=a_{1,v,1}=1$, and
 $a_{r_{v,1},w,1}=a_{2,w,1}=4$.

 \end{Example}

We start our investigation with the simplest case for completeness.
If $C$ is a $1$-dimensional binary linear code of length $n$, then $\PAut(C)$ is not generated by a fixed point free permutation of order $3$. Because
if $v$ is the non-zero vector of $C$, then $\PAut(C)$ is equal to stabilizer of $v$ which is isomorphic to $S_b \times S_{n-b}$ where
$b$ is equal to the weight of $v$.

\begin{Lemma}\label{sum}
Let $C$ be a binary linear code and $\sigma$ a fixed point free automorphism of $C$ of order $3$. Then
for any $v\in E_{\sigma}(C)$, we have that $v^{\sigma^2}+v^{\sigma}+v=0$.
\end{Lemma}

\begin{proof}
We claim that  $(v^{\sigma^2}+v^{\sigma}+v)|_{\Omega_j}=000$ for each $j=1, 2, \ldots, m$. Since $C$ is a binary code, the
definition of $E_{\sigma}(C)$ implies that $v|_{\Omega_j}$ is either a zero $3$-block or it has exactly two $1$'s. If it is a zero $3$-block,
our claim holds trivially. If $v|_{\Omega_j}$ has exactly two $1's$, then in every position of the block $\Omega_j$, exactly two of
the vectors $v^{\sigma^2}$, $v^{\sigma}$ and $v$ have value equal to $1$. So the claim holds and the proof is finished since
$\sigma$ is fixed point free.
\end{proof}
\vspace{.3cm}
\begin{Theorem}~\label{thm:2}
Let $C$ be a $k$-dimensional binary linear code where $k\geq 2$ and let $\sigma$ be a fixed point free permutation of order $3$ in ${\rm PAut}(C)$. If
${\rm dim}(E_{\sigma}(C))=2$, then there exists an involution $\theta$ in ${\rm PAut}(E_{\sigma}(C))$ which fixes every element of $F_{\sigma}(C)$.
Indeed, $\theta\in {\rm PAut}(C)$.
\end{Theorem}

\begin{proof}
Let $v$ be a non-zero vector in $E_{\sigma}(C)$. Since $v^{\sigma}\neq v$ and $E_{\sigma}(C)$ is $\sigma$-invariant,
letting $w:=v^{\sigma}$, we deduce that $\{v, w\}$ is a basis for $E_{\sigma}(C)$.

First assume that $v$ has no zero $3$-block.
Since $\sigma=(1,2,3)\ldots (n-2,n-1,n)$, and $v\in E_{\sigma}(C)$,
in each $3$-block of $v$ we have a unique zero.

 Let $\alpha_{i}=(a_{i,v,3},a_{i,w,3})$, where $a_{i,v,3}$ denotes the position of the unique zero in the $3$-block $v|_{\Omega_{i}}$, for each
 $i\in\{1,...,m\}$. Define $\theta:=\Pi_{i=1}^{m}\alpha_{i}$. The fact that $(v|_{\Omega_{i}})^{\alpha_{i}}=w|_{\Omega_{i}}$ for each $i\in\{1,...,m\}$
 implies that $v^{\theta}=w$. In particular, ${\rm PAut}(E_{\sigma}(C))$ has an involution. If $C=E_{\sigma}(C)$, we are done with the proof, so assume
 that ${\rm dim}(F_{\sigma}(C))\geq 1$. Let $z$ be an arbitrary non-zero vector in $F_{\sigma}(C)$. The definition of $F_{\sigma}(C)$ verifies that all
 $3$-blocks of $z$ are either $000$ or $111$. Since, $\theta$ stabilizes these $3$-blocks, we deduce that $z^{\theta}=z$ for any $z\in F_{\sigma}(C)$.
 Hence $\theta\in {\rm PAut}(C)$.

It should be mentioned that $v$ has a zero $3$-block if and only if $v^{\sigma}$ and $v^{\sigma^{2}}$ have a zero $3$-block at the same position. So it
is enough to concentrate on non-zero $3$-blocks of $v$ and apply the previous argument to non-zero blocks of $v$ in order to obtain an involution in
$\PAut(C)$.

\end{proof}

\begin{Corollary}~\label{cor:2}
Let $k\in \{2, 3\}$. There exists no $k$-dimensional binary linear code $C$ which satisfies
${\PAut}(C)=\langle\sigma\rangle$ where  $\sigma$ is a fixed point free permutation of order $3$.
\end{Corollary}

\begin{proof} Notice that if $C$ is a binary linear code such that $C=F_{\sigma}(C)$, then the case $\PAut(C)=\langle \sigma \rangle$ does not occur.
The rest is just a consequence of Lemma ~\ref{cor:1} and Theorem ~\ref{thm:2}.

\end{proof}

\begin{Corollary}~\label{cor:3}
Let $C$ be a $k$-dimensional binary linear code where
${\rm PAut}(C)=\langle\sigma\rangle$ and $\sigma$ is a fixed point free permutation of order $3$. Then $k\geq 4$ and ${\rm dim}(E_{\sigma}(C))\geq 4$.
\end{Corollary}

\begin{proof}
Due to Corollary ~\ref{cor:2} we come to the conclusion that $k\geq 4$. Moreover, Lemma ~\ref{cor:1} verifies that $2$ divides $\dim(E_{\sigma}(C))$.
Now, Theorem ~\ref{thm:2} guarantees that ${\rm dim}(E_{\sigma}(C))\neq 2$, hence ${\rm dim}(E_{\sigma}(C))\geq 4$.

\end{proof}

\section{Proof of the main results}

We start with recording the following elementary fact.

\begin{Lemma}~\label{lem:1}
Let $C$ be a binary linear code and $\sigma$ a fixed point free permutation of order $3$ in ${\rm PAut}(C)$. Let $D$ be a $2r$-dimensional
$\sigma$-invariant subspace of $E_{\sigma}(C)$. Then a basis of $D$ can be chosen as $\beta=\{v_{(1)},..., v_{(r)}, v_{(1)}^{\sigma},...,
v_{(r)}^{\sigma}\}$, where $v_{(i+1)}\notin \langle v_{(1)},..., v_{(i)}, v_{(1)}^{\sigma},..., v_{(i)}^{\sigma}\rangle$, for each $i\in\{1,...,r-1\}$.
\end{Lemma}

\begin{proof}

We first claim that $V:=\langle v_{(1)},..., v_{(r)}, v_{(1)}^{\sigma},..., v_{(r)}^{\sigma}\rangle$ is a $\sigma$-invariant vector space.
Let $y\in V$, then $$y=(\sum_{i=1}^r a_i v_{(i)})+ (\sum_{i=1}^r b_i v_{(i)}^{\sigma})$$ for some $a_i, b_i \in \FF_2$. It follows that
$$y^\sigma=(\sum_{i=1}^r a_i v_{(i)}^\sigma)+ (\sum_{i=1}^r b_i v_{(i)}^{\sigma^2}).$$
But by Lemma \ref{sum}, $v_{(i)}^{\sigma^2}=v_{(i)}+v_{(i)}^\sigma$, hence we get that
$$y^\sigma=(\sum_{i=1}^r b_i v_{(i)})+ (\sum_{i=1}^r (a_i+ b_i) v_{(i)}^{\sigma}),$$
so the claim follows.

We next claim that $\{v_{(1)},..., v_{(r)}, v_{(1)}^{\sigma},..., v_{(r)}^{\sigma}\}$ is a linearly independent set.
This is clear for $r=1$ since $v\neq v^{\sigma}$ for any non-zero $v\in E_{\sigma}(C)$. For $r>1$,
let $$(\sum_{i=1}^r a_i v_{(i)})+ (\sum_{i=1}^r b_i v_{(i)}^{\sigma})=0$$ for $a_i, b_i \in \FF_2$. Assume at least one of $a_r$ or $b_r$ is
non-zero. Then we have that
$$a_r v_{(r)} + b_r  v_{(r)}^{\sigma}=(\sum_{i=1}^{r-1} a_i v_{(i)})+ (\sum_{i=1}^{r-1} b_i v_{(i)}^{\sigma}). \quad (**) $$
 Let's analyse possible cases. If $a_r=1$ and $b_r=0$, then $(**)$ implies that $$v_{(r)} \in \langle v_{(1)},..., v_{(r-1)}, v_{(1)}^{\sigma},...,
 v_{(r-1)}^{\sigma}\rangle $$ a
contradiction. If $a_r=0$ and $b_r=1$, then $(**)$ implies that $$v_{(r)}^{\sigma} \in \langle v_{(1)},..., v_{(r-1)}, v_{(1)}^{\sigma},...,
v_{(r-1)}^{\sigma}\rangle$$ but by
the previous paragraph since $\langle v_{(1)},..., v_{(r-1)}, v_{(1)}^{\sigma},..., v_{(r-1)}^{\sigma}\rangle$ is $\sigma$-invariant, this is
impossible. If $a_r=1$ and $b_r=1$, then $(**)$ implies that $$v_{(r)}^{\sigma^2} =v_{(r)}+v_{(r)}^{\sigma} \in \langle v_{(1)},..., v_{(r-1)},
v_{(1)}^{\sigma},..., v_{(r-1)}^{\sigma}\rangle$$
but again $\sigma$-invariance gives a contradiction. Hence the claim is established.

Now choose a non-zero vector $v \in D$ and set $v_{(1)}:=v$. Then choose $v_{(2)}\in D$ such that $v_{(2)} \notin \langle v_{(1)},
v_{(1)}^{\sigma}\rangle$
and continue this process until $D=\langle v_{(1)},..., v_{(r)}, v_{(1)}^{\sigma},..., v_{(r)}^{\sigma}\rangle$. This finishes the proof.

\end{proof}

For a permutation automorphism $\sigma$ of $\PAut(C)$, it is easy to see that all codewords lying in the
same orbit of $\langle \sigma\rangle$ have equal weights.

\begin{Definition}
 Let $C$ be a binary linear code and $\sigma\in {\rm PAut}(C)$.
$C$ is called a distinct $\sigma$-weight code if any two $\langle \sigma \rangle$-orbits of $C$
have different weights. Otherwise, $C$ is called a non-distinct $\sigma$-weight code.

\end{Definition}

\begin{Example}
Let $C$ be a binary linear code spanned by the following vectors: $$v_{1}=111000, \ v_{2}=101101, \ v_{3}=110110.$$ Using ${\rm MAGMA}$, we find out
that $|{\rm PAut}(C)|=2^{3} \cdot 3$. Indeed, $\sigma=(1,2,3)(4,5,6)$ is in $\PAut(C)$. Moreover, $C$ is a non-distinct $\sigma$-weight code as $v_{1}$
and ${v_1+v_2}$ are in two different $\langle \sigma\rangle$-orbits but their weights are equal.

\end{Example}

 \vspace{0.3cm}

\begin{Proposition}~\label{prop:1}
Let $C$ be a binary linear code such that ${\rm PAut}(C)$ contains $\sigma$ which is a
fixed point free permutation of order $3$. If $E_{\sigma}(C)$ is a $4$-dimensional
distinct $\sigma$-weight code, then the length of $C$ is at least $30$.
\end{Proposition}

\begin{proof}
Since $4={\rm dim}(E_{\sigma}(C))$, Lemma ~\ref{lem:1} guarantees the existence of two linearly independent vectors $v$ and $w$ in $E_{\sigma}(C)$ such
that $\beta=\{v,w,v^{\sigma},w^{\sigma}\}$ is a basis for $E_{\sigma}(C)$. Let $l$ denote the number of $1$'s which overlap in
$v$ and $w$. Then it is not hard to observe that $$wt(v+w)=wt(v)+wt(w)-2l.$$
For $i\in\{0,1,2\}$, let
$$l_{i}:=|\{j\in {\rm J}_{v^{\sigma^{i}}}\cap {\rm J}_{w} :   v^{\sigma^{i}}|_{\Omega_{j}}=w|_{\Omega_{j}}\}|$$ denote the number of non-zero $3$-blocks
such that $v^{\sigma^{i}}$ and $w$ coincide. Let $s:=|{\rm J}_{w}\setminus {\rm J}_{v}|$ denote the number of non-zero $3$-blocks of $w$ such that the
associated $3$-block of $v$ is zero. In a similar way, let $t:=|{\rm J}_{v}\setminus {\rm J}_{w}|$. Let $r:=|\{j\notin {\rm J}_{v}\cup {\rm J}_{w}\}|$
be the number of common zero $3$-blocks of $v$ and $w$. One can observe that $$n=3(l_{0}+l_{1}+l_{2}+s+r+t)$$ $$wt(v)=2(l_{0}+l_{1}+l_{2})+2t$$
$$wt(w)=2(l_{0}+l_{1}+l_{2})+2s$$ $$wt(v+w)=2(l_{1}+l_{2}+t+s)$$ $$wt(v^{\sigma}+w)=2(l_{0}+l_{2}+t+s)$$ $$wt(v^{\sigma^{2}}+w)=2(l_{0}+l_{1}+t+s).$$
Moreover, as these vectors belong to distinct orbits and due to the hypothesis that $E_{\sigma}(C)$ is a distinct $\sigma$-weight code, we deduce that
$\{l_{0},l_{1},l_{2},s,t\}$ is a set which contains five distinct non-negative integers. Hence the minimal choice would be $\{0,1,2,3,4\}$ which forces
$n$ be at least $30$.
\end{proof}

The following example shows that the lower bound given in Proposition ~\ref{prop:1} for the length of such codes is attained.

\begin{Example}
Let $C$ be a binary linear code spanned by $$v_{1}= 110110110110110110110110110110,$$ $$v_{2}=110101101011011011000000000000,$$
$$v_{3}=011011011011011011011011011011,$$  $$v_{4}=011110110101101101000000000000.$$ Using ${\rm MAGMA}$, we calculate that $|{\rm PAut}(C)|=2^{15}\cdot
3^{7}$. In particular, $$\sigma=\prod_{i=0}^{9}(3i+1,3i+2,3i+3)$$ is in ${\rm PAut}(C)$ and $C=E_{\sigma}(C)$ is a $4$-dimensional binary linear code
with length $30$. Moreover, we observe that $C$ is a distinct $\sigma$-weight code.
\end{Example}

Using the notations mentioned in Proposition ~\ref{prop:1}, we state the following result.
\vspace{0.3cm}

\begin{Corollary}~\label{cor:4}
Let $C$ be a binary linear code where ${\rm PAut}(C)$ contains $\sigma$ which is a fixed point free permutation of order $3$. If $E_{\sigma}(C)$ is a
$4$-dimensional distinct $\sigma$-weight code, then $l_{i}\geq 2$ for some $i\in\{0,1,2\}$.
\end{Corollary}

\begin{proof}
This follows from the proof of Proposition ~\ref{prop:1}

\end{proof}

\begin{Theorem}~\label{thm:3}
Let $\sigma$ be a fixed point free permutation of order $3$ in $S_n$. There is no binary $4$-dimensional non-distinct $\sigma$-weight code $C$ of length
$n$ where ${\rm PAut}(C)$ is generated by $\sigma$.

\end{Theorem}

\begin{proof}

Suppose that $C$ is a binary $4$-dimensional non-distinct $\sigma$-weight code where $\PAut(C)$ is generated by $\sigma$. First of all, Corollary
~\ref{cor:3} implies that $C=E_{\sigma}(C)$. Since $4={\rm dim}(E_{\sigma}(C))={\rm dim}(C)$, Lemma ~\ref{lem:1} guarantees the existence of two
linearly independent vectors $v$ and $w$ in $E_{\sigma}(C)$ such that $\beta=\{v,w,v^{\sigma},w^{\sigma}\}$ is a basis for $E_{\sigma}(C)$. Since $C$ is
a non-distinct $\sigma$-weight code, without loss of generality, we may assume that $wt(v)=wt(w)$. We have the following cases.

\textbf{Case 1.} Assume that $wt(v)=wt(w)=2m$, where $n=3m$ is the length of $C$. This implies that neither $v$ nor $w$ has a zero $3$-block.
For $i \in \{1,...,m\}$, if $v|_{\Omega_{i}} \neq w|_{\Omega_{i}}$, let $\alpha_{i}$ be the transposition with
$(v|_{\Omega_{i}})^{\alpha_{i}}=w|_{\Omega_{i}}$ which is the transposition that interchanges the positions of the unique zero in $v|_{\Omega_{i}}$ and
$w|_{\Omega_{i}}$. If $v|_{\Omega_{i}}=w|_{\Omega_{i}}$, let $\alpha_{i}$ be an arbitrary transposition associated to the elements in that block. Let
$\alpha:=\prod_{i=1}^{m}\alpha_{i}$. Then it is easy to see that $v^{\alpha}=w$.
 Moreover, since $\sigma=\prod_{i=1}^m \sigma_i$ and each $i$ satisfies $\alpha_i^{-1} \sigma_i \alpha_i=\sigma_i^{-1}$
we have that $\alpha^{-1} \sigma\alpha=\sigma^{-1}$, so we have that
$$(v^{\sigma})^{\alpha}=w^{\sigma^{2}}\text{ and } (v^{\sigma^{2}})^{\alpha}=w^{\sigma}.$$ Therefore, $\alpha\in {\rm PAut}(C)$.

\textbf{Case 2.} Assume that $wt(v)=wt(w)=2t<2m$. From Section 2, we know that $|{\rm J}_{v}\setminus {\rm J}_{w}|=|{\rm J}_{w}\setminus {\rm J}_{v}|$.
We will construct a permutation automorphism of $C$ different from $\sigma$
via the following steps.

\begin{itemize}
\item[(1)] If $i\in {\rm J}_{v}\cap {\rm J}_{w}$, then $v|_{\Omega_{i}}$ and $w|_{\Omega_{i}}$ are both non-zero $3$-blocks. In this case, we can
    construct an involution $\alpha_{i}$ as we did in  Case 1.

\item[(2)] If $i\notin {\rm J}_{v}\cup {\rm J}_{w}$, then $v|_{\Omega_{i}}$ and $w|_{\Omega_{i}}$ are both zero $3$-blocks. In this case, let
    $\alpha_{i}$ to be an arbitrary transposition associated to the elements in block $i$.

\item[(3)] If $i\in {\rm J}_{v}\setminus {\rm J}_{w}\neq \emptyset$. Recall that $|{\rm J}_{v}\setminus {\rm J}_{w}|=|{\rm J}_{w}\setminus {\rm
    J}_{v}|$. Letting $d=|{\rm J}_{v}\setminus {\rm J}_{w}|$, then $d\geq 1$.

Let $${\rm J}_{v}\setminus {\rm J}_{w}:=\{r_{w,1},...,r_{w,d}\} \text{ and } {\rm J}_{w}\setminus {\rm J}_{v}:=\{r_{v,1},...,r_{v,d}\}.$$

We construct an involution which in particular, interchanges the  $v|_{\Omega_{r_{w,j}}}$ and $w|_{\Omega_{r_{v,j}}}$ for each $j\in \{1, \ldots,
d\}$. Namely,
let $\hat{\alpha}_{j}$ send the position of zero in the block $r_{w, j}$ of $v$  to the position of zero in the block $r_{v, j}$ of $w$, and adjust
the other
positions of these blocks so that
$\hat{\alpha}_{j}$ satisfies $$\hat{\alpha}_{j}^{-1} \sigma_{r_{w,j}}  \sigma_{r_{v,j}} \hat{\alpha}_{j} =\sigma_{r_{v,j}}^{-1} \sigma_{r_{w,j}}^{-1}
.$$

\end{itemize}

Let us define $$\alpha=\prod_{i\in {\rm J}_{v}\cap {\rm J}_{w}}\alpha_{i}\prod_{i\notin {\rm J}_{v}\cup {\rm
J}_{w}}\alpha_{i}\prod_{j\in\{1,...,d\}}\hat{\alpha}_{j}.$$
It is easy to see that $v^{\alpha}=w$. Moreover, $\alpha$ satisfies
$\alpha^{-1}\sigma\alpha=\sigma^{-1}$. As a result, we have the following equalities: $$(v^{\sigma})^{\alpha}=(v^{\alpha})^{\sigma^{2}}=w^{\sigma^{2}}$$
 $$(w^{\sigma})^{\alpha}=(w^{\alpha})^{\sigma^{2}}=v^{\sigma^{2}}$$
Hence we conclude that $\alpha\in {\rm PAut}(E_{\sigma}(C))={\rm PAut}(C)$, a contradiction.
\end{proof}

\begin{Example} Let $C$ be the code defined in Example \ref{example:1}. Recall that,
$$\sigma=(1,2,3)(4,5,6)(7,8,9)(10,11,12)(13,14,15)(16,17,18)$$ and
 $C$ is a binary linear code spanned by
$$v=110 000 101 011 000 000,$$ $$w=000 101 110 000 000 101,$$ $$v^{\sigma}=011 000 110 101 000 000,$$ $$w^{\sigma}=000 110 011 000 000 110.$$

 We have that $J_{v}\cap J_{w}=\{3\}, \ (J_{v}\cup J_{w})'=\{5\}, \ J_{v}\backslash J_{w}=\{1, 4\}, \ J_{w}\backslash J_{v}=\{2,6\}$. Moreoever,
 $$r_{w,1}=1, \ r_{w,2}=4 \text{ and } r_{v,1}=2, \ r_{v,2}=6,$$
which gives us that $\alpha_3=(8, 9), \alpha_5=(13, 14)$, $\hat{\alpha}_1=(1, 4)(2, 6)(3, 5)$, $\hat{\alpha}_2=(11, 16)(12, 18)(10, 17)$ and
it follows that $$\alpha=(1,4)(2,6)(3,5)(8,9)(10,17)(11,16)(12,18)(13,14).$$
We observe that $v^{\alpha}=w$ and $$\alpha^{-1} \sigma \alpha= (4,6,5)(1,3,2)(7,9,8)(17,16,18)(14,13,15)(11,10,12)=\sigma^{-1}.$$

 \end{Example}

\begin{Corollary}
Let $\sigma$ be a fixed point free permutation of order $3$. If $C$ is a $4$-dimensional non-distinct $\sigma$-weight code, then $S_3 \leq \PAut(C)$.

\end{Corollary}

\begin{proof}
Follows from the proof of Theorem \ref{thm:3}.
\end{proof}

\begin{Proposition}~\label{prop:2}
Let $C$ be a binary linear code and $\sigma$  a fixed point free permutation of order $3$ in ${\rm PAut}(C)$. If
$C$ is a $4$-dimensional distinct $\sigma$-weight code and $C=E_\sigma(C)$, then $\PAut(C)$ contains an involution.
\end{Proposition}

\begin{proof}
As a result of Lemma ~\ref{lem:1} and due to the hypothesis, there exists a basis $\beta=\{v,w,v^{\sigma},w^{\sigma}\}$ for $E_{\sigma}(C)$ where
$wt(v)\neq wt(w)$. Without loss of generality, we may assume that $v$ has some zero $3$-blocks.

\textbf{ Case 1.} If $v$ has at least two zero $3$-blocks, then for some distinct $i$ and $j$, we have $v|_{\Omega_{i}}=v|_{\Omega_{j}}=000$. Now we
have the following possibilities.

\begin{itemize}

\item[(1)] $w|_{\Omega_{i}},w|_{\Omega_{j}}\in\{110,011,101\}$ or $w|_{\Omega_{i}}=w|_{\Omega_{j}}=000$. In any case, there exist some positions $e$
    and $f$ in the $i$th and the $j$th $3$-blocks, respectively, for which $w_{e}=w_{f}=0$. Set $\alpha=(e,f)$. Now, we can see that
    $$v^{\alpha}=v, \quad (v^{\sigma})^{\alpha}=v^{\sigma}, \quad (v^{\sigma^{2}})^{\alpha}=v^{\sigma^{2}}$$
     Moreover, the equalities $(w^{\sigma})_{e}=(w^{\sigma})_{f}=1$ and $(w^{\sigma^{2}})_{e}=(w^{\sigma^{2}})_{f}=1$ verify that $$w^{\alpha}=w,
     \quad (w^{\sigma})^{\alpha}=w^{\sigma}, \quad (w^{\sigma^{2}})^{\alpha}=w^{\sigma^{2}}$$
     Hence $\alpha\in \PAut(C)$.

\item[(2)] $w|_{\Omega_{i}}\in\{110,011,101\}$ and $w|_{\Omega_{j}}=000$. Focus on the block $\Omega_{j}$. Since, for any $e$ and $f$ in $\Omega_{j}$
    satisfy $w_{e}=w_{f}=0$ and $v_{e}=v_{f}=0$, we can observe that the transposition $\alpha:=(e,f)$ fixes every element of $C$. Hence $\alpha\in
    \PAut(C)$.
\end{itemize}

\textbf{ Case 2.} If $v$ has a unique zero $3$-block, then Corollary ~\ref{cor:4} verifies that at least one of the elements of $\{l_{0},l_{1},l_{2}\}$
must be greater than or equal to $2$. Without loss of generality, assume that $l_{0}\geq 2$. This implies that ${\rm J}_{v}\cap {\rm J}_{w}$ includes at
least two distinct elements, say $i$ and $j$, for which $v|_{\Omega_{i}}=w|_{\Omega_{i}}$ and $v|_{\Omega_{j}}=w|_{\Omega_{j}}$. As a result,
$(v+w)|_{\Omega_{i}}=(v+w)|_{\Omega_{j}}=000$. Now $v+w$ and $w$ satisfy the conditions in (1) of Case 1.
\end{proof}

\begin{Corollary}\label{dim4}
There exists no $4$-dimensional binary linear code $C$ such that
${\rm PAut}(C)=\langle\sigma\rangle$ where $\sigma$ is a fixed point free permutation of order three.
\end{Corollary}

\begin{proof}
This is a direct consequence of Corollary ~\ref{cor:3}, 
Theorem ~\ref{thm:3}, and Proposition ~\ref{prop:2}.

\end{proof}

\begin{Corollary}
Let $C$ be a  binary linear code and $\PAut(C)=\langle \sigma \rangle$ where $\sigma$ is a fixed point free permuation
of order $3$. If $C=E_{\sigma}(C)$, then $\dim(C) \geq 6$.
\end{Corollary}

\begin{Proposition}\label{pro:313} Let $\sigma$ be a fixed point free permutation of order $3$. There is no binary $5$-dimensional linear code $C$  such
that  $E_{\sigma}(C)$ is distinct  $\sigma$-weight code and  ${\rm PAut}(C)$ is generated by $\sigma$.
\end{Proposition}

\begin{proof} Suppose that $C$ is a $5$-dimensional linear code  such that ${\rm PAut}(C)=\langle \sigma \rangle$ and $E_{\sigma}(C)$ is a distinct
$\sigma$-weight code. Then by Proposition \ref{prop:1} we have that the length of $C$ is at least $30$ and by Corollary \ref{cor:3}, we have that ${\rm
dim}(E_{\sigma}(C))=4$. On the other hand, Lemma \ref{lem:1} implies that we have a basis $\{v, w, v^{\sigma}, w^{\sigma}\}$ for  $E_{\sigma}(C)$. Let
$z$ be the basis element in $F_{\sigma}(C)$.  We use the notations in the proof of Proposition \ref{prop:1}.
Recall that $s:=|{\rm J}_{w}\setminus {\rm J}_{v}|$, $t:=|{\rm J}_{v}\setminus {\rm J}_{w}|$ and $r:=|\{j\notin {\rm J}_{v}\cup {\rm J}_{w}\}|$ where
${\rm J}_{v}$ denotes the set of indices $i$ for which $v|_{\Omega_i}$ is a non-zero $3$-block. For $i\in\{0,1,2\}$, let
$l_{i}:=|\{j\in {\rm J}_{v^{\sigma^{i}}}\cap {\rm J}_{w} :   v^{\sigma^{i}}|_{\Omega_{j}}=w|_{\Omega_{j}}\}|$ denote the number of non-zero $3$-blocks
such that $v^{\sigma^{i}}$ and $w$ coincide. We have  $n=3(l_{0}+l_{1}+l_{2}+s+r+t)$.   As $E_{\sigma}(C)$ is a distinct  $\sigma$-weight code,
by Corollary \ref{cor:4}, we have $l_i\geq 2$ for some $i\in \{0,1,2\}$.  We have two cases:

{\bf{Case 1.}} If $l_i\geq 3$ for some $i\in \{0,1,2\}$,
 without loss of generality, we can assume that $l_0 \geq 3$.
Since $l_0\geq 3$, there exists $i,j,k$ such that
$v|_{\Omega_{i}}=w|_{\Omega_{i}}, v|_{\Omega_{j}}=w|_{\Omega_{j}},v|_{\Omega_{k}}=w|_{\Omega_{k}} $.  Consider the corresponding blocks
$z|_{\Omega_{i}},z|_{\Omega_{j}}, z|_{\Omega_{k}}$. As $z\in  F_{\sigma}(C) $, at least two of the blocks $z|_{\Omega_{i}},z|_{\Omega_{j}},
z|_{\Omega_{k}}$ are identical.  Suppose $z|_{\Omega_{i}}=z|_{\Omega_{j}}$.  Consider the generator matrix $G$ of $C$.

\[ G=
\begin{bmatrix}
    z       \\
    v \\
   w \\
    v^{\sigma} \\ w^{\sigma}
\end{bmatrix}
=
\begin{bmatrix}
    ... & ... & z|_{\Omega_{i}} & \dots  &z|_{\Omega_{j }}& ... \\
    ... & .... & v|_{\Omega_{i}} & \dots  & v|_{\Omega_{j}}&... \\
    ... & ... & w|_{\Omega_{i}} & \dots & w|_{\Omega_{j}}&... \\
... & .... & v^{\sigma}|_{\Omega_{i}} & \dots  & v^{\sigma}|_{\Omega_{j}}&... \\
    ... & ... & w^{\sigma}|_{\Omega_{i}} & \dots & w^{\sigma}|_{\Omega_{j}}&...

\end{bmatrix}
\label{(1)}
\]

As $z|_{\Omega_{i}}=z|_{\Omega_{j}}$ and
$v|_{\Omega_{i}}, v|_{\Omega_{j}}, w|_{\Omega_{i}}, w|_{\Omega_{j}}, v^{\sigma}|_{\Omega_{i}}, w^{\sigma}|_{\Omega_{i}} \in \{110, 011, 101\}$
the following submatrix of $G$

\[
\begin{bmatrix}
    z|_{\Omega_{i}}   &z|_{\Omega_{j }}&  \\
    v|_{\Omega_{i}}  & v|_{\Omega_{j}}& \\
    w|_{\Omega_{i}}  & w|_{\Omega_{j}}& \\
 v^{\sigma}|_{\Omega_{i}}   & v^{\sigma}|_{\Omega_{j}}& \\
     w^{\sigma}|_{\Omega_{i}}  & w^{\sigma}|_{\Omega_{j}}&
\end{bmatrix}
\]
 has identical columns in pairs. Indeed, let $e$ be the position of zero in $v|_{\Omega_{i}}$ that is $v_e=0$ and let $f$ be the position of zero in
 $v|_{\Omega_{j}}$, that is $v_f=0$. Then $w_e=0=w_f$ and $ v^{\sigma}_e= v^{\sigma}_f=w^{\sigma}_e=w^{\sigma}_f=1$. Let $\alpha=(e,f)$.  In this case
 we have $z^{\alpha}=z, v^{\alpha}=v, w^{\alpha}=w, (v^{\sigma})^{\alpha}=v^{\sigma}, (w^{\sigma})^{\alpha}=w^{\sigma}$ . Hence $\alpha \in \PAut(C)$.

{\bf{Case 2.}}
Assume  for all $i\in\{0,1,2\}$ $l_i\leq 2$, that is $\{l_0, l_1, l_2\}=\{0,1,2\}$. Since $E_\sigma(C)$ is a distinct $\sigma$-weight code, from the
proof of Proposition \ref{prop:1}, we get that $s$ or $t$ is greater than or equal to $3$.
 That means $v$ or $w$ has at least three zero blocks. Without loss of generality assume $w$ has three zero blocks, $w|_{\Omega_i}, w|_{\Omega_j},
 w|_{\Omega_k}$.
 Then a submatrix of the generator matrix $G$ is of the form
\[
\begin{bmatrix}
    z|_{\Omega_{i}}   &z|_{\Omega_{j }}& z|_{\Omega_{k }} \\
    v|_{\Omega_{i}}  & v|_{\Omega_{j}}& v|_{\Omega_{k }} \\
    000  & 000& 000 \\
 v^{\sigma}|_{\Omega_{i}}   & v^{\sigma}|_{\Omega_{j}}& v^{\sigma}|_{\Omega_{k }} \\
     000  & 000& 000
\end{bmatrix}.
\]

By definition of $s$ and $t$,  all blocks $v|_{\Omega_{i}} , v|_{\Omega_{j}}, v|_{\Omega_{k }} $ are non-zero. Moreover, since $z\in F_\sigma(C)$, at
least two of the blocks $z|_{\Omega_{i}}, z|_{\Omega_{j}}, z|_{\Omega_{k}}$ are identical, say $z|_{\Omega_{i}}= z|_{\Omega_{j }}$. In this case, we
have a
submatrix of $G$
\[
\begin{bmatrix}
    aaa   &aaa&  \\
    v|_{\Omega_{i}}  & v|_{\Omega_{j}}& \\
    000  & 000& \\
 v^{\sigma}|_{\Omega_{i}}   & v^{\sigma}|_{\Omega_{j}}& \\
     000  & 000&
\end{bmatrix}
\]
 where $a\in \{0,1\}$. As $v|_{\Omega_{i}},  v|_{\Omega_{j}}, v^{\sigma}|_{\Omega_{i}}, v^{\sigma}|_{\Omega_{j}} \in \{110, 011, 101\}$,
there are some columns of this submatrix which are identical in pairs. Indeed, let $e$ be the position of zero in the block $v|_{\Omega_{i}} $, that is
$v_e=0$ and let $f$ be the position of zero in $v|_{\Omega_{j}} $, that is $v_f=0$. Then $v^{\sigma}_e=1=v^{\sigma}_f$. Let $\alpha=(e,f)$. It is clear
that  $\alpha$ fixes all basis elements and hence $\alpha \in \PAut(C)$.

Since in all possible cases, we deduce that there is an involution in $\PAut(C)$, the proof is finished.
\end{proof}

\begin{Proposition}\label{pro:314}Let $\sigma$ be a fixed point free permutation of order $3$ in $S_n$. There is no binary $5$-dimensional linear code
$C$ of length $n\geq 30$ such that  $E_{\sigma}(C)$ is a non-distinct  $\sigma$-weight code and  ${\rm PAut}(C)$ is generated by $\sigma$.
\end{Proposition}

\begin{proof}
Let $\sigma$ be a fixed point free permutation of order $3$ in $S_n$ where $n\geq 30$.  Suppose that $C$ is a $5$-dimensional linear code such that
${\rm PAut}(C)=\langle \sigma \rangle$ and $E_{\sigma}(C)$ is non-distinct  $\sigma$-weight code. By Corollary \ref{cor:3} and  by Lemma \ref{lem:1}, we
have a basis $\{v, w, v^{\sigma}, w^{\sigma}\}$ for  $E_{\sigma}(C)$.  Since  $E_{\sigma}(C)$ is a non-distinct  $\sigma$-weight code, without loss of
generality we can assume that $wt(v)=wt(w)$. Let $z$ be the basis element for $F_{\sigma}(C)$. Let $l_0, l_1, l_2, r, s, t$ be the set of integers as
defined in the proof of  Proposition \ref{prop:1}. Recall that $n=3(l_0+l_1+l_2+s+t+r)$. As  $wt(v)=wt(w)$, it is clear that $s=t$. If $s=t\geq 3$, then
we can construct an involution which fixes all basis elements as in the Case 2 of the proof of Proposition \ref{pro:313}. There are linear  codes  such
that  $l_i\geq 3$ for some $i\in \{0,1,2\}$, if this is the case  we can construct an involution which fixes all basis elements as in the Case 1  of the
proof of Proposition \ref{pro:313}.

It should be mentioned that  there are some linear codes with the property that  $l_i\leq 2$ for all $i$ and $s=t\leq 2$. So we should consider the case
$r\geq 1$, that is the number of common zero blocks of $v$ and $w$ is at least one. Let $v|_{\Omega_i}=w|_{\Omega_i}=000$. The corresponding block of
$z$ is $z|_{\Omega_i}\in \{000,111\}$.  In fact we have a submatrix of $G$ in the form
\[
\begin{bmatrix}
    aaa   \\
    000 \\
   000 \\
 000 \\
     000
\end{bmatrix}
\]
where $a\in \{0,1\}$.
So any transposition on this block (any transposition changing the columns of the above submatrix) fixes all basis elements.

In all cases, we construct an involution in $\PAut(C)$, so such a $C$ does not exist.
\end{proof}

\begin{proof}[Proof of Theorem \ref{main1}]
Follows from  Corollary \ref{cor:3}, Corollary \ref{dim4}, Proposition \ref{pro:313} and Proposition \ref{pro:314}.

\end{proof}

\begin{proof}[Proof of Theorem \ref{main2}] Follows from Corollary \ref{cor:3}, Corollary \ref{dim4}.

\end{proof}

\begin{Remark}\label{rem1}

Whenever $E_{\sigma}(C)$ is a non-distinct $\sigma$-weight code with length $n< 30$, we can not construct transpositions as in the proof of Proposition
\ref{pro:314}. For example, for the length $n=6$, let $\sigma=(1,2,3)(4,5,6)$ and  consider a binary linear code with the following generator matrix

\[ G=
\begin{bmatrix}
    z       \\
    v \\
v^{\sigma}\\
   w
     \\ w^{\sigma}
\end{bmatrix}
=
\begin{bmatrix}
    111&000&  \\
    110  & 000&  \\
 011  & 000&
     \\

000  & 101 & \\
     000  & 110&
\end{bmatrix}.
\]

In this case there is no transposition or involution fixing all basis elements.
 However, the involution $\alpha=(1,3)(5,6)$ satisfies $z^{\alpha}=z$,
$v^{\alpha}=v^{\sigma}$,${v^{\sigma}}^{\alpha}=v$ $w^{\alpha}=w^{\sigma}$,  ${w^{\sigma}}^{\alpha}=w$. Hence $\alpha \in \PAut(C)$. Moreover
calculations in MAGMA give that $\PAut(C)\cong S_3\times S_3$

Now let $n=9$. Let $\sigma=(1,2,3)(4,5,6)(7,8,9)$ and consider the  code generated by the following matrix
\[ G=
\begin{bmatrix}
    z       \\
    v \\
 v^{\sigma}\\
   w \\
     w^{\sigma}
\end{bmatrix}
=
\begin{bmatrix}
    111  &000&  111&\\
    110  & 000&  101& \\
 011   & 000& 110 &\\
    000  & 101 &   110 &\\

     000  & 110 & 011&
\end{bmatrix}.
\]

In this case  $\alpha=(1,2)(5,6)(7,9)$  acts on the basis elements as follows
$$z^{\alpha}=z, \ v^{\alpha}=v,\ w^{\alpha}=w^{\sigma}, \ {(v^{\sigma})}^{\alpha}=v+ v^{\sigma}, \ {(w^{\sigma})}^{\alpha}=w.$$  We observe from the
calculations carried out in MAGMA that $\tau=(1,3)(4,5)(8,9)\in \PAut(C)$ and $K=\langle (1,9),(2,7),(3,8)\rangle$ is an elementary abelian $2$-subgroup of
order $8$ of ${\rm PAut}(C)$. Let
$L= K \rtimes {\langle \sigma \rangle} $ then
${\rm PAut}(C) = L \rtimes {\langle \tau \rangle} $.

For $n=12$,  let $\sigma=(1,2,3)(4,5,6)(7,8,9)(10,11,12)$ and consider the  code generated by the following matrix
\[ G=
\begin{bmatrix}
    z       \\
    v \\
 v^{\sigma}\\
   w \\
     w^{\sigma}
\end{bmatrix}
=
\begin{bmatrix}
    111  & 000&  111& 111\\
    110  & 000&  101& 000\\
    011  & 000&  110& 000\\
    000  & 110&  101& 101   \\
    000  & 011&  110& 110
\end{bmatrix}.
\]
 The involution $\alpha=(1,3)(4,6)(8,9)(11,12)$ acts on the basis elements as
$$v^{\alpha}=v^{\sigma},~~{v^{\sigma}}^{\alpha}=v,~~w^{\alpha}=w^{\sigma},~~{w^{\sigma}}^{\alpha}=w, z^{\alpha}=z.$$ So $\alpha \in \PAut(C)$ and
calculations in MAGMA give that$|\PAut(C)|=12$.

\end{Remark}

\begin{Remark}\label{Calc}
The examples in Remark \ref{rem1} show that under the assumption when $n< 30$, we can not find a pattern to construct an involution in $\PAut(C)$.

For $5$-dimensional cubic codes of length $n$, we obtain some computational results for the case where $n\in\{6,9,12,15,18,21\}$. These results are
listed in Table~\ref{tab:1}. Indeed, using the construction methods in \cite{LingSole}(see Theorem 5.1 and section VI, the case $m=3$), we construct all
cubic codes of length $n$, where $n\in\{6,9,12,15,18,21\}$ and determine the order of the permutation automorphism groups of these codes.  

For $n\in \{24, 27\}$, with respect to our computer facilities, the calculations seem to take years. Up to date, we get some partial results as follows:
\begin{itemize}
\item For the length $n=24$, with respect to construction methods in \cite{LingSole}, there are  $8206520925$ cubic codes and up to date, we get
    $5559$ different cubic codes up to equivalence. Among these  cubic codes, the smallest order for the permutation automorphism group is $48$ and
    the group structure for these permutation automorphism groups of order $48$ is $C_2\times S_4$.

\item For the length $n=27$, with respect to construction methods in \cite{LingSole}, there are $263945834061$ cubic codes and up to date, we get
    $12789$ different cubic codes up to equivalence. Among these cubic codes, the smallest order for the permutation automorphism is $48$ and the
    group structure for these permutation automorphism groups of order $48$ is $C_2\times S_4$.

\item For the length $n\in\{24,27\}$, these partial results support the fact that a cubic code of dimension $5$
and length $n\in \{24,27\}$ has permutation automorphism group of order greater than $3$.
All computations are done in SAGEMATH.
\end{itemize}

\end{Remark}

\begin{table}[h]
\begin{tabular}{ |c|c|c|c|c|c|c| }
 \hline
 \multirow{1}{15em}{\small{Length of Cubic Codes of dimension five }}& \small{6}& \small{9}&\small{12}&\small{15}&\small{18}&\small{21}
\\&&&&&&\\

 \hline

 \multirow{2}{15em}{\small{The number of different Cubic Codes up to equivalence}}&\small{2} &\small{15}&\small{67}&\small{244}&\small{765}&\small{$
 2142$}\\&&&&&&\\

 \hline

\multirow{2}{15em}{\small{Minimum order of the permutation automorphism groups}}&
\small{36}&\small{12}&\small{6}&\small{6}&\small{6}&\small{12}\\&&&&&&\\

 \hline

 \multirow{2}{15em}{\small{Structure of automorphism groups with minimum order}}&\small{ $S_3\times
 S_3$}&\small{$D_{12}$}&\small{$S_3$}&\small{$S_3$}&\small{$S_3$}&\small{$D_{12}$}\\&&&&&&\\

 \hline

\end{tabular}
\vspace{.3cm}
\caption{Order of Permutation Automorphism Groups of $5$-dimensional\\  Cubic Codes of Length$\leq 21$}
~\label{tab:1}
\end{table}

\begin{Remark}

There are examples of cubic codes such that $\PAut(C)=\langle \sigma \rangle$. In the next section, we manage to construct such codes whenever
$\dim(C)\geq 6$. See  Example \ref{Dim6length18}, Example \ref{Dim10length30} and Example \ref{Dim18length36}.
\end{Remark}

\section{Some results on higher dimensional cubic codes}
In this section, we provide some generalizations of the ideas that are used in the previous sections. We construct some cubic codes of dimension greater than or equal to  $6$ having permutation automorphism group of order $3$. This bound on the dimension emphasizes the importance of the results for dimension less than or equal to $ 5$.
During this progress, on one side, we present some results for the structure of higher dimensional cubic codes and on the other side, we apply these
general ideas to give some conditions on the structure of the putative extremal  binary self-dual $[72, 36, 16]$ code having permutation automorphism
group of order $3$.

Let $C$ be a binary linear code of length $n=3m$ and $$\sigma=(1,2,3)(4,5,6)\dots (n-2,n-1,n)$$ be a fixed point free
automorphism in $\PAut(C)$.

\begin{Definition}   A rectangular block diagonal matrix is a $rk \times sk$ rectangular matrix
$${\rm diag}(D_1,D_2,\dots, D_k)=\left[
\begin{array}{cccccc}
  D_1 & \bigzero & \bigzero& \bigzero&\dots&\bigzero \\

  \bigzero & D_2&\bigzero&\bigzero&\dots &\bigzero\\

\vdots &\vdots& \ddots &\vdots&\vdots&\vdots \\

\bigzero & \bigzero&\dots& D_i& \dots&\bigzero  \\

\vdots &\vdots& \ddots &\vdots&\ddots&\vdots \\

\bigzero & \bigzero&\dots& \bigzero & \dots& D_k
\end{array}
\right]$$
where $k,r,s$ are positive integers and  all $D_i$'s on the diagonal block are $r\times s$ matrices and $\bf{0}$'s are $r\times s$ zero matrices.
\end{Definition}

Consider the following matrices:
$$E_0=\bf{0}=\begin{bmatrix}
    0  &0&0&\\
    0  &0&0&
\end{bmatrix},~ E_1=\begin{bmatrix}
    1  &0&1&\\
    0 &1&1&
\end{bmatrix},~E_2=\begin{bmatrix}
    1  &1&0&\\
    1 &0&1&
\end{bmatrix},~ E_3=\begin{bmatrix}
    0  &1&1&\\
    1  &1&0&
\end{bmatrix}.
$$

\begin{Lemma}\label{rref} Let $C$ be a binary linear code of length $3m$ such that $\sigma\in {\PAut}(C)$. If $\dim (E_{\sigma}(C))=2k$, then
$E_{\sigma}(C)$ has a generating matrix $G$ such that, up to permutation equivalence, the row reduced echelon form of $G$ is in the form
$\left[
\begin{array}{c|c}
\bf{E}& \bf{M} \end{array}\right]$,
where  ${\bf E}={\rm diag}(E_1,E_1,\dots, E_1)$ is a  $2k\times 3k$ rectangular block diagonal matrix
 and
 ${\bf{M}}$ is a $2k\times 3(m-k)$ block matrix whose blocks are in
$\{E_0,E_1,E_2,E_3\}$.

\end{Lemma}
\begin{proof} Assume that $C$ is a binary linear code with $\sigma\in \PAut(C)$. If ${\rm dim}(E_{\sigma}(C))=2k$, then Lemma \ref{lem:1} guarantees the
existence of a basis for $E_{\sigma}(C)$ of the following form $$\{v_1, v_1^{\sigma}, v_2, v_2^{\sigma}, \dots, v_k, v_k^{\sigma} \}.$$
Let $G$ denote the associated generating matrix of $E_{\sigma}(C)$. Thus $G$ has the following form

\[G=\begin{bmatrix}
    v_1       \\
    v_1^{\sigma} \\
v_2\\
v_2^{\sigma}\\
\vdots\\
   v_k \\
    v_k^{\sigma}
\end{bmatrix}.
\]

Observe that, for $i\in\{1,2,\dots, k\}$ and $j\in\{1,2,\dots, m\}$ we have $${v_i|}_{\Omega_j}, {v_i^{\sigma}|}_{\Omega_j}  \in \{101, 110,
011,000\}.$$
 
Up to permutation equivalence we can choose the first row $v_1$ so that $v_1|_{\Omega_1}=101$. 
Then $v_1^{\sigma}|_{\Omega_1}=110$. By adding the first row two the second row we get the first block as $E_1$.  For $j\in \{2,3,4,\dots, k\}$, we have ${v_j|}_{\Omega_1},{v_j^{\sigma}|}_{\Omega_1}\in\{v_1|_{\Omega_1}, v_1^{\sigma}|_{\Omega_1}, v_1^{\sigma^2}|_{\Omega_1}, 000 \} $. Appropriate elementary row operations result in the appearance of  zero blocks in the first $3$-block columns.  If $v_2|_{\Omega_2}=000$, we can consider an appropriate permutation which changes the places of second $3$-block columns with a non-zero $3$-block columns so that for this equivalent code, we have $v_2|_{\Omega_2}\neq 000$.   If $v_2|_{\Omega_2}= 101$ then apply the same procedure as in the first part. If $v_2|_{\Omega_2}\neq 101$, then $v_2|_{\Omega_2}+v_2^{\sigma}|_{\Omega_2}= 101$. Now apply the procedure in the first part. The rest follows from the same procedure and elementary row operations. Note that whenever we apply that appropriate  permutation to $E_{\sigma}(C)$, we also apply it to $F_{\sigma}(C)$ subcode to get a code which is permutation equivalent to $C$.

\end{proof}

\begin{Example}  Consider the binary linear code $C$ generated by the following matrix\[ 
\begin{bmatrix}
    z      \\
    v_1 \\
 v_1^{\sigma}\\
   v_2 \\
     v_2^{\sigma}
\end{bmatrix}
=
\begin{bmatrix}
    111  & 000&  000& 111 & 000\\
    000  & 110&  101& 000& 101\\
    000  & 011&  110& 000& 110\\
    000  & 101&  101& 101 & 011 \\
    000  & 110&  110& 110& 101
\end{bmatrix}.
\]
For the basis  elements $v_1, v_2$ of $E_{\sigma}(C)$, we have ${v_1|}_{\Omega_1}={v_2|}_{\Omega_1}=000$. So by applying the permutation $(1,13)(2,14)(3,15)$ to the code, we get the equivalent code generated by the matrix
\[\begin{bmatrix}
    000  & 000&  000& 111 & 111\\
    101  & 110&  101& 000& 000\\
    110  & 011&  110& 000& 000\\
    011  & 101&  101& 101 & 000 \\
    101  & 110&  110& 110& 000
\end{bmatrix}.\] Now we replace the code $C$ by the new equivalent one. Consider the generating matrix of $E_{\sigma}(C)$ seperately. For the matrix
\[G=\begin{bmatrix}
    
    101  & 110&  101& 000& 000\\
    110  & 011&  110& 000& 000\\
    011  & 101&  101& 101 & 000 \\
    101  & 110&  110& 110& 000
\end{bmatrix},\] by adding first row to the second row we get 

\[\begin{bmatrix}
    101  & 110&  101& 000& 000\\
    011  & 101&  011& 000& 000\\
    011  & 101&  101& 101 & 000 \\
    101  & 110&  110& 110& 000
\end{bmatrix}.\] By adding first row to the fourth one and second row to the third one we obtain the matrix 

\[\begin{bmatrix}
    
    101  & 110&  101& 000& 000\\
    011  & 101&  011& 000& 000\\
    000  & 000&  110& 101 & 000 \\
    000  & 000&  011& 110& 000
\end{bmatrix}.\] 
Now in this matrix $v_2|_{\Omega_2}=000.$
So apply the permutation $(4,10)(5,11)(6,12)$ to get the code generated by the matrix
\[\begin{bmatrix}
    101  & 000&  101& 110& 000\\
    011  & 000&  011& 101& 000\\
    000  & 101&  110& 000 & 000 \\
    000  & 110&  011& 000& 000
\end{bmatrix}.\] By adding the third row to the fourth one we get 
\[\begin{bmatrix}
    101  & 000&  101& 110& 000\\
    011  & 000&  011& 101& 000\\
    000  & 101&  110& 000 & 000 \\
    000  & 011&  101& 000& 000
\end{bmatrix}\] which is in the form
$\left[
\begin{array}{c|c}
\bf{E}& \bf{M} \end{array}\right]$.

Observe that the codes generated by the following generating matrices are permutation equivalent 
\[\begin{bmatrix}
    111  & 000&  000& 111 & 000\\
    000  & 110&  101& 000& 101\\
    000  & 011&  110& 000& 110\\
    000  & 101&  101& 101 & 011 \\
    000  & 110&  110& 110& 101
\end{bmatrix},~~
\begin{bmatrix}
    000   & 111& 000& 000 &111\\
    101  & 000&  101& 110& 000\\
    110  & 000&  110& 011& 000\\
    011  & 101&  101& 101 & 000 \\
    101  & 110&  110& 110& 000
\end{bmatrix}.
\] via the permutation $\gamma=(1,13)(2,14)(3,15)(4,10)(5,11)(6,12)$ and 
the matrices 
\[\begin{bmatrix}
    000   & 111& 000& 000 &111\\
    101  & 000&  101& 110& 000\\
    110  & 000&  110& 011& 000\\
    011  & 101&  101& 101 & 000 \\
    101  & 110&  110& 110& 000
\end{bmatrix}, 
\begin{bmatrix}
    000 & 111 & 000 &000& 111\\
    101  & 000&  101& 110& 000\\
    011  & 000&  011& 101& 000\\
    000  & 101&  110& 000 & 000 \\
    000  & 011&  101& 000& 000
\end{bmatrix}
\] are row equivalent.

\end{Example}
\begin{Proposition}\label{pro:3k2k} Let $C$ be a binary linear code of length $3k\geq 3$ and $\sigma \in {\PAut}(C)$. If ${\rm dim}(E_{\sigma}(C))=2k$,
then ${\PAut}(C)\neq \langle \sigma \rangle $.

\end{Proposition}

\begin{proof} 
If $C=E_{\sigma}(C)$, then
Lemma \ref{rref} assures that ${\bf E}={\rm diag}(E_1,E_1,\dots, E_1)$ is a generating matrix of $C$.  For each $i\in \{1,2,\dots, k\}$, let
$\alpha_i=(3i-2,3i-1)$. Obviously, the transposition $\alpha_{i}$ changes the places of zeros in $i$th $E_1$ on ${\bf E}={\rm diag}(E_1,E_1,\dots, E_1)$ . 
 Set $\alpha=\prod_{i=1}^{k} \alpha_i$. It is clear that $\alpha \in {\PAut}(C)$.
  On the other hand, if $C\neq E_{\sigma}(C)$, then we have  $\dim (F_{\sigma}(C))\geq 1$.  Consider the involution $\alpha$ defined in the previous
  case.  Clearly $\alpha$ fixes each element in $F_{\sigma}(C)$ which implies that $\alpha \in{\PAut}(C)$.
\end{proof}

Consider the above-mentioned matrix $\bf{M}$. Let $M_i$ represent the $i$th $3$-block column of $\bf{M}$ so that we can rewrite $\bf{M}$ in  the following form:

$$\bf{M}=\left[
\begin{array}{c|c|c|c}
M_1&M_2&\dots& M_{m-k} \end{array}\right]$$

By definition of $\bf{M}$, each $M_i$ has the following form:
\[M_i=\begin{bmatrix}
    E_{i,1}   \\
    E_{i,2}\\
\vdots\\
   E_{i,k}
\end{bmatrix}
\] where $E_{i,j}\in \{E_0,E_1,E_2,E_3\}$, for each $j\in\{1,...,k\}$.

We progress by considering the following hypotheses on the structure of the matrix ${\bf M}$ which cover all possible cases.

{\bf{Hypothesis A}}: For each $i\in\{1,2,\dots, m-k\}$, the non-zero blocks of $M_i$ are all equal.

{\bf{Hypothesis B}}: For some $i\in\{1,2,\dots, m-k\}$, there exist $r,s\in\{1,...,k\}$ such that $E_{i,r}\neq E_{i,s}$ on  $M_i$.

\begin{Example}

$${\bf{M}}=\left[ \begin{array}{cccc}
  E_1 & \bigzero & \bigzero& \bigzero\\

  \bigzero & \bigzero&\bigzero&E_2\\

\bigzero & E_3&\bigzero&\bigzero\\

\bigzero & \bigzero&E_1 &\bigzero
\end{array}
\right],~~
{\bf{M'}}=\left[ \begin{array}{cccc}
  E_1 & \bigzero & \bigzero& \bigzero\\

  \bigzero & \bigzero&\bigzero&E_2\\

\bigzero & E_3&\bigzero&E_2\\

\bigzero & E_3&E_2 &\bigzero
\end{array}
\right],~~
{\bf{M''}}=\left[ \begin{array}{cccc}
  E_1 & E_3 & E_1& \bigzero\\

  E_2 & \bigzero&\bigzero&E_2\\

\bigzero & E_2&E_1&E_3\\

\bigzero & \bigzero&E_1 &\bigzero
\end{array}
\right]$$
Observe that the matrices $\bf{M}$ and $\bf{M'}$ satisfy Hypothesis A while the matrix $\bf{M''}$ satisfies Hypothesis B.
\end{Example}

\begin{Proposition}\label{permM} Let $C$ be a binary linear code of length $3m\geq 3$ and $\sigma \in {\PAut}(C)$. Assume that ${\rm
dim}(E_{\sigma}(C))=2k$. With the above notations, if the matrix ${\bf M}$ in the generating matrix  $\left[ \begin{array}{c|c}
\bf{E}&\bf{M}\end{array}\right]$ satisfies Hypothesis A, then ${\PAut}(C)\neq \langle \sigma \rangle $.

\end{Proposition}
\begin{proof} Consider the generating matrix $\left[
\begin{array}{c|c}
\bf{E}&\bf{M}\end{array}\right]$ for $E_{\sigma}(C)$. We can define $\alpha=\prod_{i=1}^k\alpha_i$ on the matrix $\bf{E}$ in a similar way to the one
we defined earlier in the proof of Proposition \ref{pro:3k2k}. As ${\bf M}$ satisfies Hypothesis A, for each $i\in \{1,2\dots m-k\}$, whenever $M_{i}$
has a non-zero block $E_{i,j}$, then all the other non-zero blocks which appear in $M_{i}$ are equal to $E_{i,j}$. With this observation, whensoever
$M_{i}$ is a non-zero matrix, let $\beta_i$ be the transposition which interchanges the positions of zeros in the non-zero block $E_{i,j}$ of $M_{i}$.
Instead, if for some $l\in \{1,2,\dots, m-k\}$, the matrix $M_{i}$ has no non-zero block, then let $\beta_l$ be a transposition related to the column
numbers occupied by $M_{i}$ in the matrix ${\bf M}$. Finally, define $\beta=\prod_{i=1}^{m-k}\beta_i$. Apparently, $\alpha$ and $\beta$ are disjoint
involutions which implies that $\alpha\beta$ is an involution in $ \PAut (E_{\sigma}(C))$. Moreover, $\alpha\beta$ fixes each element of
$F_{\sigma}(C)$, henceforth $\alpha\beta \in \PAut(C)$ and ${\PAut}(C)\neq \langle \sigma \rangle $.
\end{proof}
\begin{Corollary}\label{ConditionB} Let $C$ be a binary linear code such that $\PAut(C)=\langle \sigma \rangle$. With the above notations, the matrix
${\bf{M}}$ in the generating matrix
$\left[
\begin{array}{c|c}
\bf{E}& \bf{M} \end{array}\right]$ of $E_{\sigma}(C)$ satisfies Hypothesis B.

\end{Corollary}
\begin{proof}
This is a direct consequence of Proposition ~\ref{permM}.
\end{proof}

\begin{Example} Let $C$ be a binary linear code such that $\sigma\in \PAut(C)$. Assume that $E_{\sigma}(C)$ has the generating matrix:
$$\left[ \begin{array}{ccccccc}
  E_1 & \bigzero & \bigzero& \bigzero&\bigzero& E_2 &\bigzero\\

  \bigzero & E_1&\bigzero&\bigzero&E_1&\bigzero&\bigzero\\

\bigzero & \bigzero&E_1&\bigzero&\bigzero &\bigzero& E_3\\

\bigzero & \bigzero&\bigzero & E_1&\bigzero&\bigzero&\bigzero
\end{array}
\right]=\left[ \begin{array}{c|c|c|c|c|c|c}
  101 & 000 & 000& 000&000& 101 &000\\
  011 & 000 & 000& 000&000& 011 &000\\
  \hline
  000 & 101&000&000&011&000&000\\
  000 & 011&000&000&110&000&000 \\

\hline
000 &000&101&000&000 &000& 110\\
000 &000&011&000&000 &000& 101\\
\hline

000 &000&000 & 101&000&000&000\\
000 &000&000 & 011&000&000&000
\end{array}
\right]$$

With the notations mentioned in the proof of Proposition ~\ref{permM}, observe that $\alpha_1=(1,2)$, $\alpha_2=(4,5)$, $\alpha_3=(7,8)$, and
$\alpha_{4}=(10,11)$. Additionally, $\beta_1=(13,15)$, $\beta_2=(16, 17)$, and $\beta_3=(20,21)$. Therefore,
$$\alpha\beta=(1,2)(4,5)(7,8)(10,11)(13,15)(16,17)(20,21)$$
Clearly, $\alpha\beta$ fixes any element of $F_{\sigma}(C)$, hence it is an element of $\PAut(C)$.
\end{Example}

\begin{Example} Let $C$ be a linear code such that $\sigma\in \PAut(C)$. Let $E_{\sigma}(C)$ has the following generating matrix
$$\left[ \begin{array}{ccccccc}
  E_1 & \bigzero & \bigzero& \bigzero&\bigzero& E_1 &\bigzero\\

  \bigzero & E_1&\bigzero&\bigzero&E_3&\bigzero&E_2\\

\bigzero & \bigzero&E_1&\bigzero&\bigzero &E_1& E_2\\

\bigzero & \bigzero&\bigzero & E_1&\bigzero&\bigzero&E_2
\end{array}
\right]=\left[ \begin{array}{c|c|c|c|c|c|c}
  101 & 000 & 000& 000&000& 011 &000\\
  011 & 000 & 000& 000&000& 110 &000\\
  \hline
  000 & 101&000&000&110&000&101\\
  000 & 011&000&000&101&000&011 \\

\hline
000 &000&101&000&000 &011& 101\\
000 &000&011&000&000 &110& 011\\
\hline

000 &000&000 & 101&000&000&101\\
000 &000&000 & 011&000&000&011
\end{array}
\right]$$
 Then $\alpha_1=(1,2), ~ \alpha_2=(4,5),~ \alpha_3=(7,8), ~\alpha_{4}=(10,11)$.
$\beta_1=(14,15), \beta_2=(16, 18), \beta_3=(19,20)$. Thus
$$\alpha\beta=(1,2)(4,5)(7,8)(10,11)(14,15)(16,18)(19,20)$$
It is clear that $\alpha\beta$ fixes any element of $F_{\sigma}(C)$ and hence it is an element of $\PAut(C)$.
\end{Example}

Hereafter, we give some examples of binary cubic codes such that the permutation automorphism groups have order three.

\begin{Example}\label{Dim6length18} Consider the matrix

$$\left[ \begin{array}{cccccc}
000& 111& 111&  000& 000& 111\\
111&  111& 000&  000& 111& 000\\

  E_1&\bigzero &E_2 & E_3 &E_2& \bigzero\\

  \bigzero & E_1&E_2&\bigzero&E_1&E_2\\

\end{array}
\right] =\left[ \begin{array}{c|c|c|c|c|c}
000& 111& 111&  000& 000& 111\\
111&  111& 000&  000& 111& 000\\
\hline
110& 000& 101&  110& 101& 000 \\
011& 000& 110&  011& 110&  000\\
\hline
000& 110& 101&  000& 011&  101\\
000& 011& 110&  000& 101&  110
\end{array}
\right].$$

It is straightforward to observe that this matrix generates a $6$-dimensional binary linear code $C$ of length $18$. Calculations in MAGMA show that
$\PAut(C)=\langle \sigma \rangle$. Note that the matrix $\bf{M}$ in $\left[
\begin{array}{c|c}
\bf{E}& \bf{M} \end{array}\right]$ satisfies Hypothesis B. Moreover, $E_{\sigma}(C)$ is a non-distinct $\sigma$-weight code.

\end{Example}

\begin{Example} \label{Dim10length30} Let $C$ be the binary linear $[30,10]$ code generated by the following matrix:

$$\left[ \begin{array}{cccccccccc}
  E_1&\bigzero &\bigzero&\bigzero&\bigzero&E_3 & E_1 &E_2& E_3&E_2\\

  \bigzero & E_1&\bigzero&\bigzero&\bigzero&E_2&E_1&E_1&E_1&E_2\\

\bigzero & \bigzero&E_1&\bigzero&\bigzero &E_1& E_1&E_3&E_1&\bigzero\\

\bigzero & \bigzero&\bigzero & E_1&\bigzero&E_2&E_1&E_2&E_1&E_1\\

\bigzero & \bigzero&\bigzero &\bigzero&E_1& \bigzero &E_1&E_2&E_1&E_1
\end{array}
\right]$$ which is explicitly equal to the matrix $$\left[ \begin{array}{c|c|c|c|c|c|c|c|c|c}
  101 & 000 & 000& 000&000& 110 &011&101&110&101\\
  011 & 000 & 000& 000&000& 101 &110&011&101&011\\
  \hline
  000 & 101&000&000&000&  101&011&011&011&101\\
  000 & 011&000&000&000&  011&110&110&110&011 \\

\hline
000 &000&101&000&000   &011&011&110&011&000 \\
000 &000&011&000&000   &110&110&101&110&000\\
\hline

000 &000&000 & 101&000  &101&011&101&011&011\\
000 &000&000 & 011&000  &011&110&011&101&110\\

\hline

000 &000&000 &000 &101&000  &011&101&011&011\\
000 &000&000 &000 &011&000  &110&011&101&110
\end{array}
\right].$$

One can observe that $\sigma\in \PAut(C)$ and $C=E_{\sigma}(C)$. For this cubic code, calculations in MAGMA show that
 $\PAut(C)=\langle \sigma \rangle$. The matrix $\bf{M}$ in $\left[
\begin{array}{c|c}
\bf{E}& \bf{M} \end{array}\right]$ satisfies Hypothesis B.
Moreover, $E_{\sigma}(C)$ is a non-distinct $\sigma$-weight code.

\end{Example}

\begin{Example}\label{Dim18length36}
Assume that $C$ is the binary linear $[36,18]$ code generated by the following matrix:

$$\left[ \begin{array}{cccccccccccc}
111&000 &111&111&000&111&000& 111 &000& 111&111&111\\

111& 000& 111 &000 &111&  000& 111& 111& 000& 111& 111& 111\\

000& 000& 000& 000& 000&  000& 111& 111&111& 000& 111& 000\\

111& 111 &111& 111& 111&  111& 000& 111& 000& 111& 000& 000\\

111& 000& 111&111& 000&  111& 111& 000& 000&111& 111& 000\\

111& 111& 000&000& 111& 111& 111& 000& 111& 000& 111& 111\\

000& 111& 111& 000& 111&  111& 000& 111& 000& 111& 000& 111\\

  E_1&\bigzero &\bigzero&\bigzero&\bigzero&\bigzero&E_3 & E_1 &\bigzero& E_3&E_2&E_2\\

  \bigzero & E_1&\bigzero&\bigzero&\bigzero&\bigzero&E_2&E_1&E_1&E_1&E_2&E_2\\

\bigzero & \bigzero&E_1&\bigzero&\bigzero &\bigzero &E_1& E_1&E_3&E_1&\bigzero&\bigzero\\

\bigzero & \bigzero&\bigzero & E_1&\bigzero&\bigzero &E_2&E_1&E_2&E_1&E_1&E_1\\

\bigzero & \bigzero&\bigzero &\bigzero&E_1& \bigzero& \bigzero  &E_1&E_2&E_1&E_1&\bigzero\\

\bigzero & \bigzero&\bigzero &\bigzero&\bigzero& E_1 &E_1&E_1&E_2&E_1&E_1&E_2
\end{array}
\right]$$

Computations in MAGMA show that we have $\PAut(C)=\langle \sigma \rangle$. Moreover, $E_{\sigma}(C)$ is a non-distinct $\sigma$-weight code.
\end{Example}

Considering the features obtained in the previous examples, stating the following conjecture seems logical.

\begin{Conjecture} Let $C$ be a binary linear code such that $\sigma \in \PAut(C)$. If $\PAut(C)=\langle \sigma \rangle$, then $E_{\sigma}(C)$ is a
non-distinct $\sigma$-weight code.

\end{Conjecture}

\section{Concluding Results}

\begin{Theorem} Let $C$ be a binary linear self-dual $[72,36,16]$ code. With the notations mentioned in the previous sections, if $ \PAut(C)=\langle
\sigma \rangle$, then the matrix ${\bf{M}}$ in the generating matrix
$\left[
\begin{array}{c|c}
\bf{E}& \bf{M} \end{array}\right]$ of $E_{\sigma}(C)$ satisfies Hypothesis B.

\end{Theorem}
\begin{proof} Follows from Corollary \ref{ConditionB}.
\end{proof}
\begin{Remark}
It is worth mentioning that if $C$ is a self-dual binary linear code and $\sigma\in\PAut(C)$, then $E_{\sigma}(C)$ and $F_{\sigma}(C)$ are
self-orthogonal.

Whenever we want to construct a self-orthogonal $E_{\sigma}(C)$ for which the matrix ${\bf M}$ in  $\left[
\begin{array}{c|c}
\bf{E}& \bf{M} \end{array}\right]$ satisfies Hypothesis B, we obtain very large permutation automorphism groups.
\end{Remark}

\begin{Example}
The matrix
$$\left[ \begin{array}{cccccccccc}
111& 111& 111& 111& 111& 111& 111& 111& 111& 111\\

111& 111& 111& 111& 111& 111& 111& 111& 000& 000\\

111& 111& 111& 111& 111& 111& 000& 000& 000& 000\\

000& 000& 000& 000& 111& 111& 111& 111& 111& 111\\

111& 111& 000& 000& 000& 000& 111& 111& 000& 000\\

  E_1&\bigzero &\bigzero&\bigzero&\bigzero&E_3&E_3 & E_1 &E_1& E_2\\

  \bigzero & E_1&\bigzero&\bigzero&\bigzero&E_2&E_2&E_3&E_2&E_2\\

\bigzero & \bigzero&E_1&\bigzero&\bigzero &E_3 &E_1& E_3&\bigzero&\bigzero\\

\bigzero & \bigzero&\bigzero & E_1&\bigzero&E_2 &E_2&E_3&E_1&E_3\\

\bigzero & \bigzero&\bigzero &\bigzero&E_1& E_2& E_1  &E_1&\bigzero&\bigzero

\end{array}
\right]$$
generates a self-dual code $C$ of length $30$ and dimension $15$. Calculations in MAGMA show that the order of $\PAut(C)$ is $49152$.

\end{Example}
\begin{Example}
The matrix
$$\left[ \begin{array}{cccccccccc}
111& 111& 111& 111& 111& 111& 111& 111& 111& 111\\

111& 111& 000& 000& 000& 000& 000& 000& 000& 000\\

000& 000& 111& 111& 000& 000& 000& 000& 000& 000\\

111& 111& 000& 000& 111& 111& 000& 000& 000& 000\\

111& 111& 000& 000& 000& 000& 111& 000& 000& 111\\

  E_1&\bigzero &\bigzero&\bigzero&\bigzero&E_3&E_3 & E_1 &E_1& E_2\\

  \bigzero & E_1&\bigzero&\bigzero&\bigzero&E_2&E_2&E_3&E_2&E_2\\

\bigzero & \bigzero&E_1&\bigzero&\bigzero &E_3 &E_1& E_3&\bigzero&\bigzero\\

\bigzero & \bigzero&\bigzero & E_1&\bigzero&E_2 &E_2&E_3&E_1&E_3\\

\bigzero & \bigzero&\bigzero &\bigzero&E_1& E_2& E_1  &E_1&\bigzero&\bigzero

\end{array}
\right]$$
generates a binary self dual code  of length $30$ and of dimension $15$. Calculations in MAGMA show that order of
$\PAut(C)$ is $1440$.
\end{Example}

All these information motivate us to present the following conjecture:
\begin{Conjecture}
If $C$ is  a binary linear self-dual $[72,36,16]$ code such that $\sigma \in \PAut(C)$, then $ \PAut(C)\neq\langle \sigma \rangle$.

\end{Conjecture}

\section*{Acknowledgements}
{\.I}. TUVAY was partially supported by the Scientific and Technological Research Council of Turkey (T{\"u}bitak) through the research program
Bideb-2219.

The numerical calculations reported in this paper were partially performed at TUBITAK ULAKBIM, High Performance and Grid Computing Center (TRUBA
resources).

\end{document}